\newcommand{\T}{\Theta}
\newcommand{\thetahat}{\hat\theta}

\documentclass[journal]{IEEEtran}
\ifCLASSINFOpdf
\else
\fi
\usepackage[ruled]{algorithm2e}
\usepackage{amsmath}
\usepackage{amssymb}
\usepackage{xcolor}
\usepackage{amsthm}
\usepackage{graphicx}
\usepackage{stfloats}
\usepackage{comment}
\usepackage{subcaption}

\usepackage[T1]{fontenc}
\usepackage[font=small,labelfont=bf,tableposition=top]{caption}

\DeclareCaptionLabelFormat{andtable}{#1~#2  \&  \tablename~\thetable}

\newtheorem{definition}{Definition}

\newtheorem{lemma}{Lemma}
\newtheorem{remark}{Remark}
\newtheorem{assumption}{Assumption}
\newtheorem{proposition}{Proposition}

\begin{document}
%
% paper title
% Titles are generally capitalized except for words such as a, an, and, as,
% at, but, by, for, in, nor, of, on, or, the, to and up, which are usually
% not capitalized unless they are the first or last word of the title.
% Linebreaks \\ can be used within to get better formatting as desired.
% Do not put math or special symbols in the title.
\title{Incentive Mechanism in the Sponsored Content Market with Network Effects}
%
%
% author names and IEEE memberships
% note positions of commas and nonbreaking spaces ( ~ ) LaTeX will not break
% a structure at a ~ so this keeps an author's name from being broken across
% two lines.
% use \thanks{} to gain access to the first footnote area
% a separate \thanks must be used for each paragraph as LaTeX2e's \thanks
% was not built to handle multiple paragraphs
%
\author{Mina~Montazeri, Pegah~Rokhforoz, Hamed~Kebriaei,~\IEEEmembership{Senior Member,~IEEE}, and
        Olga~Fink,~\IEEEmembership{Member,~IEEE}% <-this % stops a space
\thanks{Mina Montazeri and Hamed Kebriaei  are with the School of ECE, College of Engineering, University of Tehran, Tehran, Iran. Emails: (mina.montazeri@ut.ac.ir, kebriaei@ut.ac.ir). Pegah Rokhforoz is with the Chair of Intelligent Maintenance Systems, ETH Zurich, Switzerland. Email: (foroz@ibi.baug.ethz.ch). Olga Fink is with the Laboratory of Intelligent Maintenance and Operation Systems, EPFL, Switzerland. Email: (olga.fink@epfl.ch).  }% <-this % stops a space
%\thanks{J. Doe and J. Doe are with Anonymous University.}% <-this % stops a space
%\thanks{Manuscript received April 19, 2005; revised August 26, 2015.}
}

\maketitle
% As a general rule, do not put math, special symbols or citations
% in the abstract or keywords.
\begin{abstract}
We propose an incentive mechanism for the sponsored content provider market in which the communication of users can be represented by a graph and the private information of the users is assumed to have a continuous distribution function. %The induced mechanism is between users and the content provider.
The content provider stipulates incentive rewards to encourage users to reveal their private information truthfully and increase their content demand, which leads to an increase in the advertising revenue.
We prove that all users gain a non-negative utility and disclose their private information truthfully.
%These two important properties of the induced mechanism are known as individual rationality and incentive compatibility.
Moreover, we study the effectiveness and scalability of the proposed mechanism in a case study with different network structures. 
%The results demonstrate that the proposed mechanism satisfies individual rationality and incentive compatibility. In addition, the user who has the most influence on other users in a social network has the most impact on the utility of the content provider as well which means that his/her untruthful data decreases the utility of content provider more.
\end{abstract}

% Note that keywords are not normally used for peerreview papers.
\begin{IEEEkeywords}
Sponsored content market, mechanism design, continuous private information,  network system.
\end{IEEEkeywords}

% For peer review papers, you can put extra information on the cover
% page as needed:
% \ifCLASSOPTIONpeerreview
% \begin{center} \bfseries EDICS Category: 3-BBND \end{center}
% \fi
%
% For peerreview papers, this IEEEtran command inserts a page break and
% creates the second title. It will be ignored for other modes.
\IEEEpeerreviewmaketitle

\section{Introduction}
Smartphones facilitate people's interaction to share information with their friends, which leads to a huge cellular data flow. Due to the increasing use of cellular data services, the data cost would increase and become one of the critical concerns of the users. This increasing data cost results in decreasing the data consumption in social networks \cite{yang2017noncooperative}, content markets \cite{xiong2020contract}, crowdsourcing, and crowdsensing tasks \cite{wang2019optimization,hu2020blockchain}. For this reason, the content/service providers provide some incentives to support partially the mobile users’ data usage \cite{xiong2020contract,wang2019optimization,hu2020blockchain}. In this paper, we focus on the problem of designing incentive rewards in addressing the sponsored content market \cite{xiong2017economic}. 

{The concept of sponsored content market has been introduced in setups where the Content Provider (CP) designs an incentive mechanism by providing some incentive rewards to motivate users to consume more sponsored content  and hence maximize his revenue by displaying more advertisements \cite{xiong2017economic}.  % In other words, by providing some incentive rewards to users, the content demand is increased and also the revenue for the sponsored advertisements is increased, resulting in an increase in the overall profit of the CP.
However, designing optimal incentive rewards in this market can be challenging, and has been the subject of numerous studies \cite{xiong2020contract,nie2018stackelberg}. Previous research studies focusing on mechanism design for sponsored contend markets had several limiting assumptions: such as the assumption that the CP has  full knowledge of the users' information  \cite{xiong2017economic} or that the users will truthfully report their private information to the CP. However, this assumption is not realistic in real-world scenarios since users have
incentive to report their information incorrectly to gain
more profit. If the users are not revealing information truthfully, it is difficult for the CP to design
an optimal incentive reward. To address this challenge, contract theory offers
a framework to design mechanisms under asymmetric information in which, the mechanism designer is not aware of the user’ private information \cite{nisanalgorithmic}. In most of the previous studies that applied contract theory to content provider market, the authors assumed that the users' private
information is drawn from a discrete distribution which is not realistic since it cannot capture a wide range of possible values \cite{zhang2017survey,xiong2020contract}. Moreover, previous research did not consider the influence of user interactions on each other. This poses again a limitation to its applicability in real world scenarios since in many social network systems, such as Facebook or Twitter, the behaviors of the users do influence each other.

To address the outlined limitations, in this paper, we propose an optimal incentive mechanism for the content provider market when the users have private information. The proposed incentive mechanism aims to achieve the following three goals: 1) Maximize the content provider's utility function; 2) Guarantee the participation of the users (individual rationality);  3) Motivate users to share their private information truthfully (incentive compatibility). 
To accurately reflect real setups in sponsored content markets, we assume that the users can influence each other's behavior through social interactions that are represented by a graph. Moreover, we assume that {the private information of the users, which is also referred to as the "type" of the users, can be modeled by a continuous distribution function. This private information %(user's type) 
implies the strength of social network ties between the users.}}

{We model the sponsored content market as an optimization problem in which the CP aims to obtain incentive rewards which maximize his utility and also satisfy the incentive compatibility and individual rationality constraints.} The proposed mechanism leads to a tractable constrained functional optimization which links the demand and the incentive reward function. We prove that the proposed mechanism fulfills the incentive compatibility and individual rationality properties. In addition, we evaluate the effectiveness and scalability of the proposed mechanism in a case study in which we consider different numerical examples with different network structures. In summary, the main contributions of the paper are as follows:
\begin{itemize}
    \item  We propose an incentive mechanism for users whose social interaction is represented by a graph with continuous private information.
    \item We reformulate the mechanism design problem to a tractable optimization which links the incentive reward to the demand function.
    \item We prove that the proposed mechanism fulfills the individual rationality and incentive compatibility properties.
\end{itemize}
{
\section{Related work}
Previous research studies have focused on different aspects of designing optimal incentive rewards in the content provider market.
%Therefore, it is important to design mechanisms for optimally allocating incentive rewards in the sponsored content market. Different approaches have been proposed to tackle this challenge.
One of the previously proposed  approaches for sponsored content market is based on (Bayesian) Stackelberg games. In \cite{xiong2017economic}, the interaction among the network operator, the content provider, and the users is formulated as a Stackelberg game. In this setup, the CP aims to design incentive rewards that maximize his utility based on the assumption that the CP knows all the information of the users. However, this assumption is not realistic in real-world scenarios. To address this limitation, \cite{nie2018stackelberg} used the Bayesian Stackelberg games to allocate incentive rewards in the sponsored content market. Bayesian Stackelberg games are particularly applicable in setups where the CP is uncertain about some parameters in the objective function of the users. Even though \cite{nie2018stackelberg} assumes that the CP has incomplete information about the agents, which is more realistic compared to previous research, the proposed method still has several unrealistic assumptions which preclude its application in real world scenarios. Firstly, the authors assume that all the users voluntarily participate in
the market. Secondly, in the Bayesian Stackelberg game, the optimal strategy of the CP is obtained by the average with respect to the distribution of the user's private information. This means that all users receive the same amount of reward which is again unrealistic since users prefer to receive a reward based on their own utility functions. 

Building upon the efforts to tackle the challenge of designing rewards based on users' utility function, several recent research studies have focused on designing an optimal incentive reward that encourages users to report their private information truthfully. These efforts have been seen in various applications, including resource/task allocation markets \cite{roughgarden2010algorithmic} and blockchain-based environments \cite{cai2020truth,cai2022truthful}.
%Designing the incentive reward function associated with the private information of each user and also encouraging users to report their private information truthfully is a topic that has been studied in various applications in recent years has been recognized in many research studies on resource/task allocation markets and blockchain-based environments.
%One example of this line of research is the voting system in blockchain in which voters can vote based on their true opinion even if it is a minority opinion \cite{cai2020truth,cai2022truthful}. 
To encourage the truthful sharing, contract theory offers a framework
to design incentive reward under asymmetric information in such a way that it motivates users to report their private information truthfully\cite{roughgarden2010algorithmic}.
Several research studies have applied contract theory to the content provider market
\cite{zhang2017survey,xiong2020contract}. However,
these studies assumed that the users' private information follows a discrete distribution function. While this assumption simplifies the incentive mechanism, it is a restrictive assumption because in many applications the private information of the users follows a continuous distribution \cite{mirrlees1997information}. The only study that did take continuous distribution into consideration, proposed an incentive mechanism for the sponsored content market where users have continuous private information \cite{montazeri2021optimal}. Yet, the study did not consider the influence of user interactions on each other. This poses again a limitation to its applicability in real world scenarios since in many social network systems, such as Facebook or Twitter the behaviors of the users do influence each other. In such systems, the users' interactions can be modeled as a graph where the nodes represent the users and the edges represent the influence strength of the users' social ties \cite{jadbabaie2019optimal}. However, the interactions between users have not been considered in any of the previous works that applied contract theory in the content provider market.}

%To the best of our knowledge, this is the first time that an incentive mechanism design problem has been proposed for content provider where users social interactions are represented by a graph and the users are provided with incentives to share their private information.

\section{Model and Problem Formulation}
\label{Model and Problem Formulation}
{
\subsection{Model of Sponsored Content Markets}
We model the sponsored content platform as a market consisting of two entities: a CP and a set of users $\mathcal{N}=\{1,2,\dots,N\}$. % who have a private information. To accurately reflect real setups in sponsored content markets, we assume that the users can influence each other's behaviour through social interactions that are represented by a graph.
To accurately reflect real setups in sponsored content markets, we assume that the users can influence each other's behavior through social interactions that are represented by a graph. However, based on the user's personality, the strength of the network effect for each user is different. We model the strength of the network effect for user $i$ by parameter $\theta_i\in\T_{i}$, where $\T_{i}=[\underline{\theta}, \bar{\theta}]$. This parameter is private information of each user and is  neither known to the content provider nor to the other users. To realistically depict reality, private information is modeled as a continuous variable as it offers more versatility and can encompass a broader range of values \cite{montazeri2022distributed}.

In this market, the CP first obtains the content demand $x_{i}:{\Theta}\to{\mathbb{R}^+}$ and incentive reward $R_{i}:{\Theta}\to{\mathbb{R}^+}$, where $\T=\T_{1}\times\cdots\times\T_{N}$ as function of users' private information. This incentive reward encourages users to consume more content.
After that users share their private information, denoted as $\thetahat_{i}\in\T_{i}$, with the CP such that they maximize their profit. In general, their shared private information may not necessarily be their actual private information, i.e. $\theta_i$, unless it is intrinsically optimal for them to share it truthfully. The general schematic of the sponsored content market is shown in Figure~\ref{Fig:system_model}.
}
\begin{figure}[h!]
	\centering
	\includegraphics[width=1\linewidth]{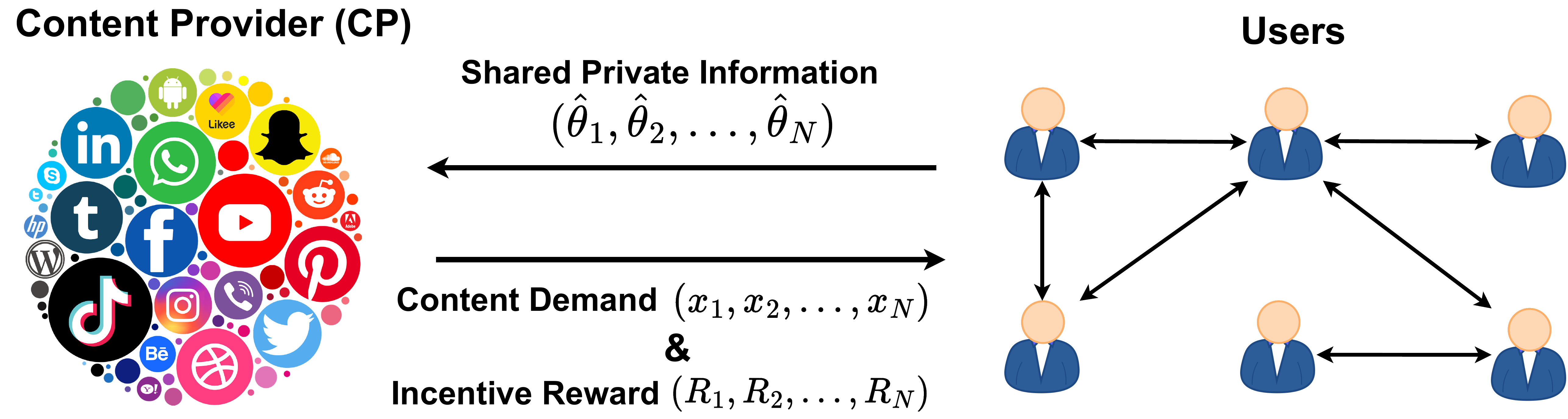}
	\caption{The schematic of sponsored content market
%	for the users whose social connections are represented by a graph. The CP obtains the content demand and incentive reward functions. Then, the users determine their optimal announcement private information based on the CP strategy.
}
	\label{Fig:system_model}
\end{figure}

\subsection{Problem Formulation}
\label{Model}
 {
Consider agent $i$ who is embedded in a given social network and participate in the sponsored content market. In this case, %its utility function includes four terms: internal utility $\psi(.)$, external utility $\sigma(.)$, cost of using content demand $P(.)$ and the reward received from CP $R(.)$ \cite{candogan2012optimal,xiong2019dynamic}.
%%%
%In the following,
we define the user's utility as follows \cite{xiong2019dynamic}:}
\begin{equation}
\begin{aligned}
  \label{eq:utility_NU_1}
&U_{i}(x(\thetahat),R_i(\thetahat),\theta_i)=
\psi(x_i(\thetahat))+\theta_i\sigma(x(\thetahat))-px_i(\thetahat)+R_i(\thetahat),
\end{aligned}
\end{equation}
\normalsize
where $\thetahat={\{\thetahat_{1},\cdots,\thetahat_{N}}\}$ and $x(\thetahat)={\{x_{1}(\thetahat),\cdots,x_{N}(\thetahat)}\}$ are a shared private information profile and content demand profile. These include the shared private information and content demand of each agent, respectively. 
 {The first term $\psi(x_i(\thetahat))$  represents the internal utility that user $i$ gains from consuming and enjoying the content demands. This term derives from user content consumption, independent of the consumption of user's neighbors.
%models the satisfaction of user $i$ upon consuming the content.
Inspired by \cite{candogan2012optimal,jadbabaie2019optimal}, we model this term as a linear-quadratic function: $\psi(x_i(\thetahat)) = ax_i(\thetahat) - \frac{b}{2}x_i^2(\thetahat)$, where $a\geq 0$ models the maximum internal demand willingness rate, and $b\geq 0$ models the willingness elasticity factor \cite{candogan2012optimal}.
We model the second term, which represents the external utility due to the network effects, %exhibited by the sponsored content platform
as $\sigma(x(\thetahat))= x_i(\thetahat)\sum_{j=1}^{N} g_{ij}x_j(\thetahat)$  \cite{candogan2012optimal,jadbabaie2019optimal}. In this formulation, $g_{ij} \ge 0$ represents the influence strength of the social tie of user $i$ on user $j$. Thus, the users’ behaviors in terms of content demand are influencing each other. As mentioned before, the parameter $ \theta_i$, which is the private information of user $i$, controls the strength of the network effect for user $i$ and is known neither to the CP nor to the other users. However, it can be assumed that it is commonly known that $\theta_i$ follows a distribution $F(\theta_i)$. Note that this assumption is reasonable since
the CP can estimate the statistical information about the distribution of users' private information by learning from user historical behavior or conducting a user survey. We also assume that $F(\theta_i)$ is
continuously differentiable.}
The term $p x_i(\thetahat)$ indicates the cost that user $i$ has to pay to the mobile operator for the consumption of $x_i$.
%%%%%%%%%%%%%%%%%%%%%%%%%%%%%%%%%%%%%

 The CP's utility comprises two parts: the total advertisement revenue gained from the user's content consumption and the total rewards paid to all the users. 
 Thus, the CP's utility can be formulated as:
\begin{align}
\label{eq:utility_CP}
U^{CP}(x(\thetahat),R(\thetahat))=\mathbf{E}_{\hat{\theta} }[\sum_{i\in\mathcal{N}}\big (Q(x_i(\thetahat))-R_i(\thetahat)\big)],
\end{align}
where $R(\cdot)={\{R_{1}(\cdot),\cdots,R_{N}(\cdot)}\}$.
The function $Q$ is the advertisement revenue
gained from data usage of user $i$ and is given by $Q(x_i(\thetahat)) = sx_i(\thetahat) - \frac{t}{2}x_i^2(\thetahat)$ in which $s, t \geq 0$ are predefined coefficients characterizing the extent of the concavity of the function \cite{candogan2012optimal}. 
Thus, the CP offers $(x(\cdot), R(\cdot))$, which provides the content demand consumption and incentive reward for users based on users' shared private information.% The steps of this mechanism is shown in Figure \ref{Fig:system_model}.\\
\subsection{Designing the Incentive Mechanism}
\label{Problem Formulation}
We address the mechanism design problem from the content provider perspective.
The goal of the CP is to design a utility-maximizing mechanism which receives shared private information from the users and determines both the incentive rewards and the demand consumption.
 Solving this problem presents two main challenges and requires the CP to ensure that the users participate in the market and share their private information truthfully. 
In the following, we provide the definition of these properties.
\begin{definition}%(Individual Rationality)
	\label{Def:IR}
	A mechanism is individually rational if the users gain a non-negative utility by sharing their private information truthfully, i.e.
	\begin{align}
		\label{Eq:IR_cons}
	IR_{\theta_{i}}:& \mathbf{E}_{\hat{\theta}_{-i}}	[U_{i}(x_i({\theta_i,\hat{\theta}_{-i}}),x_{-i}({\theta_i,\hat{\theta}_{-i}}),R_i({\theta_i,\hat{\theta}_{-i}}),\theta_i)] \ge 0. %\\
	%\nonumber
%	&  \forall \theta \in  [\underline{\theta}, \bar{\theta}]^N,
	\end{align}
	where $x_{-i}(\cdot)={\{x_{1}(\cdot),\cdots,x_{i-1}(\cdot),x_{i+1}(\cdot),\cdots,x_{N}(\cdot)}\}$ is the content demand of all users except user $i$. A similar definition holds for $\thetahat_{-i}={\{\thetahat_{1},\cdots,\thetahat_{i-1},\thetahat_{i+1},\cdots,\thetahat_{N}}\}$.
%	$\thetahat_{-i}={\{\thetahat_{1},\cdots,\thetahat_{i-1},\thetahat_{i+1},\cdots,\thetahat_{N}}\}$ and $x_{-i}(\cdot)={\{x_{1}(\cdot),\cdots,x_{i-1}(\cdot),x_{i+1}(\cdot),\cdots,x_{N}(\cdot)}\}$.
\end{definition}
\begin{definition}%(Incentive Compatibility)
	\label{Def:IC}
	A mechanism is incentive-compatible %(IC) 
	if the users achieve an equal or higher utility by sharing their private information truthfully, i.e.
	\begin{align}
	\label{Eq:IC_cons}
	IC_{\theta_{i},\hat{\theta}_{i}}:& \mathbf{E}_{\hat{\theta}_{-i}}[U_{i}(x_i({\theta_i,\hat{\theta}_{-i}}),x_{-i}({\theta_i,\hat{\theta}_{-i}}),R_i({\theta_i,\hat{\theta}_{-i}}),\theta_i)] \ge\\
	\nonumber
	&\mathbf{E}_{\hat{\theta}_{-i}}[ U_{i}(x_i({\hat{\theta_i},\hat{\theta}_{-i}}),x_{-i}({\hat{\theta_i},\hat{\theta}_{-i}}),R_i({\hat{\theta_i},\hat{\theta}_{-i}}),\theta_i)].
%	\nonumber
%	&  \forall \hat{\theta}, \theta \in  [\underline{\theta}, \bar{\theta}]^N.
	\end{align}
\end{definition}
%Notice that, Definitions \ref{Eq:IR_cons} and \ref{Eq:IC_cons} express the notions of \emph{ex interim} individual rationality and \emph{Bayesian} incentive compatibility, respectively.
%From the mechanism design point of view, the CP should maximize his/her utility, subject to $IR_{\theta_{i}}$ and $IC_{\theta_{i},\hat{\theta}_{i}}$ constraints. Therefore, the optimal mechanism, can be obtained by solving the following maximization problem:
%We propose a  mechanism, in which the CP determines demand consumption and incentive rewards associated to each NU in order to achieve an optimal utility. 
%
The CP designs
the demand consumption and incentive reward functions such that the three objectives are achieved simultaneously: (i) motivate users to participate in the market, (ii) ensure that users truthfully share their private information, (iii) maximize the CP's utility. Thus, the
optimal mechanism is obtained by solving the following
optimization problem:
\begin{align}
\label{opt_1}
&\max_{\{R(\cdot), x(\cdot)\}} U^{CP}(x(\thetahat),R(\thetahat)) \quad\textnormal{s.t.} \,\,\quad IR_{\theta_{i}}, \quad IC_{\theta_{i},\hat{\theta}_{i}}.%\quad
%\forall \hat{\theta_i}, \theta_i \in  [\underline{\theta}, \bar{\theta}]^N
\end{align}
Since the optimization \eqref{opt_1} is not convex and also not straightforward to solve, we propose in the following how we can reformulate it in order to find optimal decisions for the CP.
\section{Solution of the Mechanism}
\label{Sec:Mechanism}
 In this section, we use the following propositions to investigate
 the solution of the optimization problem \eqref{opt_1}.
For ease of presentation, we define:
\begin{equation}
\begin{aligned}
\label{taghyir_motoghayer}
V_i(\hat{\theta}_i)\equiv &\mathbf{E}_{\hat{\theta}_{-i}}[(a-p)x_i(\hat{\theta_i},\hat{\theta}_{-i})-\frac{b}{2}x_i^2(\hat{\theta_i},\hat{\theta}_{-i}),\\
\gamma_i(\hat{\theta}_i)\equiv& \mathbf{E}_{\hat{\theta}_{-i}}[ x_i(\hat{\theta_i},\hat{\theta}_{-i})\sum_{j\in\mathcal{N}} g_{ij}x_j(\hat{\theta_i},\hat{\theta}_{-i})],\\
r_i(\hat{\theta}_i)\equiv &\mathbf{E}_{\hat{\theta}_{-i}}[R_i(\hat{\theta_i},\hat{\theta}_{-i})],\\
C_i(\hat{\theta_i})\equiv &\mathbf{E}_{\hat{\theta}_{-i}}[sx_i(\hat{\theta_i},\hat{\theta}_{-i}) - \frac{t}{2}x_i^2(\hat{\theta_i},\hat{\theta}_{-i})],\\
\tilde{U}_{i}(\theta_i,\hat{\theta_i})&\equiv
\mathbf{E}_{\hat{\theta}_{-i}}[ U_{i}(x_i({\hat{\theta_i},\hat{\theta}_{-i}}),x_{-i}({\hat{\theta_i},\hat{\theta}_{-i}}),R_i({\hat{\theta_i},\hat{\theta}_{-i}}),\theta_i)].%,i\in\mathcal{N}, \theta_i \in [\underline{\theta}. \bar{\theta}].
\end{aligned}
\end{equation}
This results in $\tilde{U}_{i}(\theta_i,\hat{\theta}_i)=
V_i(\hat{\theta}_i)+\theta_i\gamma_i(\hat{\theta}_i)+r_i(\hat{\theta}_i)$.
\begin{comment}
\begin{equation*}
\tilde{U}_{i}=
V_i(\hat{\theta_i})+\theta_i\gamma_i(\hat{\theta_i})+r_i(\hat{\theta_i}).
\end{equation*}
\end{comment}
\begin{proposition}
	\label{Lem1}
	In optimization \eqref{opt_1}, if the $IR_{\theta_{i}}$ constraint is satisfied for $\underline{\theta}$ and the $IC_{\theta_{i},\hat{\theta}_{i}}$ constraint is satisfied as well, for all $\hat{\theta_i}, \theta_i$, then the $IR_{\theta_{i}}$ is satisfied for all $ \theta_i $. Also, in any optimal solution, we have $\tilde{U}_{i}(\underline{\theta},\underline{\theta})=0$, i.e., the $IR_{\theta_{i}}$ constraint is active for $\underline{\theta}$.	%%%%%%%%%%%%%%%%%%%%%%%
%	In \eqref{opt_1}, the $IR_{\theta_{i}}$ constraint is satisfied for all $\theta$ if $IR_{\underline{\theta_i}}$ is binding (or active) which means:
%	\begin{equation*}
%	\tilde{U}^{NU}_{i}(\underline{\theta}_i,\underline{\theta_i})=0.
%	\end{equation*}
	%	\begin{equation}
	%\mathbf{E}_{\theta_{-i}}[U_{NU_i}(\underline{\the%ta_i},x_i({\underline{\theta_i},\theta_{-i}}),x_{%-i}({\underline{\theta_i},\theta_{-i}}),R_i({\und%erline{\theta_i},\theta_{-i}}))] = 0.
	%	\end{equation}
\end{proposition}
\begin{proof}
{
%%%%%%%%%%%%%%%%
%	Using $IC_{\theta_i,\underline{\theta}}$ and $\frac{d \tilde{U}^{NU}_{i}}{d \theta_i} \ge 0$, we get:
Let us use the following inequality:
	\begin{align}
	\label{Eq:IR_pr3}
	\tilde{U}_{i}(\theta_i,\theta_i) \ge \tilde{U}_{i}(\theta_i,\underline{\theta}) \ge \tilde{U}_{i}(\underline{\theta},\underline{\theta}),
	\end{align}
	where the first inequality holds as the result of $IC_{\theta_i,\underline{\theta}}$ and the second inequality follows from $\frac{d \tilde{U}_{i}(\cdot)}{d \theta_i} \ge 0$. Thus, it follows $\tilde{U}_{i}(\theta_i,\theta_i) \ge\tilde{U}_{i}(\underline{\theta},\underline{\theta}) $.
Hence, we conclude that $IR_{\underline{\theta}}$ implies $IR_{{\theta_i}}$. To finalize the proof, it is left to demonstrate that $IR_{\underline{\theta}}$ should be binding. If $IR_{\underline{\theta}}$ is not bind, there exists $\epsilon > 0$ which could decrease  $R_i(\theta)$ such that all constraints of \eqref{opt_1} are satisfied and the utility of the CP is also increased.
}
\end{proof}
The previous proposition suggests that we can safely remove all the $IR_{{\theta_i}}$ constraints in which $\theta_i \neq \underline{\theta}$. In other words, the infinite number of inequality constraints of $IR_{{\theta_i}}$ can be converted to a single equality constraint.
Next, we provide the following equivalent formulation of the %optimization
problem \eqref{opt_1}:
\begin{subequations}
	\label{opt_2}
	\begin{align}
\label{cost}
&\max_{\{R(\cdot), x(\cdot)\}} U^{CP}(x(\thetahat),R(\thetahat))\\
s.t.\quad
\label{cons1}
	&\gamma_i ' (\theta_i) \ge 0 %, \; \; \; \; \; \; \; \; \; \forall \theta_i \in [\underline{\theta_i}, \bar{\theta_i}],i \in \mathcal{N}
	, \\
\label{cons2}
		&r_{i}({\theta_{i}})=   [\int_{\underline{\theta}}^{\theta_i} \gamma_i  (y) dy] -\theta_i \gamma_i(\theta_i)- V_i(\theta_i).
\end{align}
\end{subequations}
\begin{proposition}
	\label{Lem2}
	Optimizations \eqref{opt_1} and  \eqref{opt_2} are equivalent.
\end{proposition}
\begin{proof}
		To show the equivalence between the two optimization problems, it is sufficient to show that for each optimal solution of optimization \eqref{opt_1}, there exists a solution of optimization \eqref{opt_2} with the same objective, and vice versa. %We divide the proof of this proposition into two parts.
		First, we show that given a solution to optimization \eqref{opt_1}, we can find a solution to optimization \eqref{opt_2} with the same objective. If Equation \eqref{cons2} is satisfied, 
 we can write $IC_{\theta_i,\hat{\theta_i}}$ as:
	\begin{align*}
	%	\label{IC_1}
	&V_i({\theta_i})+\theta_i\gamma_i({\theta_i})+ [\int_{\underline{\theta}}^{\theta_i}\gamma_i  (y) dy] -\theta_i \gamma_i(\theta_i)- V_i(\theta_i)\ge\\
	\nonumber
	&V_i(\hat{\theta_i})+\theta_i\gamma_i(\hat{\theta_i})+ [\int_{\underline{\theta}}^{\hat{\theta_i}}\gamma_i  (y) dy] -\hat{\theta_i} \gamma_i(\hat{\theta_i})- V_i(\hat{\theta_i}).
	\end{align*}
	%%%%%%%%%%%%%%%%%%%%%%%%%%%%%%%%%%%%%%%
	%x_i({\hat{\theta_i},\hat{\theta}_{-i}})
	%%%%%%%%%%%%%%%%%%%%%%%
	Thus, we have:
		\begin{equation}
		\begin{cases}
		\label{IC_3}
\int_{\hat{\theta_i}}^{{\theta_i}}\gamma_i(y)dy \ge
	(\theta-\hat{\theta})\gamma_i(\hat{\theta_i}),&\theta_i > \hat{\theta_i},\\
	(\hat{\theta}-{\theta})\gamma_i(\hat{\theta_i})\ge
	\int_{{\theta_i}}^{\hat{\theta_i}}\gamma_i(y)dy,&\theta_i < \hat{\theta_i}.
	\end{cases}
	\end{equation}
	%%%%%%%%%%%%%%%%%%%%%%%
\begin{comment}
	For $\theta_i > \hat{\theta_i}$, we get:
	\begin{align}
	\label{IC_3}
	\int_{\hat{\theta_i}}^{{\theta_i}}\gamma_i(y)dy \ge
	(\theta-\hat{\theta})\gamma_i(\hat{\theta_i}),
	\end{align}	
	and	for $\theta_i < \hat{\theta_i}$, we get:
	\begin{align}
	\label{IC_4}
	(\hat{\theta}-{\theta})\gamma_i(\hat{\theta_i})\ge
	\int_{{\theta_i}}^{\hat{\theta_i}}\gamma_i(y)dy.
	\end{align}	
	\end{comment}
	Regarding \eqref{cons1}, both equations in \eqref{IC_3} hold true.
Next, we show that given an optimal solution to optimization \eqref{opt_1}, we can find a solution to optimization \eqref{opt_2} with the same objective.
As a first step, we prove that truthfulness implies monotonicity of $\gamma_i({\theta_i})$. According to $IC_{\theta_i,\hat{\theta_i}}$ and $IC_{\hat{\theta_i},{\theta_i}}$, we obtain:
	\begin{align}
	\nonumber
	&V_i({\theta_i})+\theta_i\gamma_i({\theta_i})+r_i({\theta_i})\ge V_i(\hat{\theta_i})+\theta_i\gamma_i(\hat{\theta_i})+r_i(\hat{\theta_i}), \\
	&V_i(\hat{\theta_i})+\hat{\theta_i}\gamma_i(\hat{\theta_i})+r_i(\hat{\theta_i})\ge V_i({\theta_i})+\hat{\theta_i}\gamma_i({\theta_i})+r_i({\theta_i}).
	\label{IC_6}
	\end{align}
	By summing these two inequalities of Equation \eqref{IC_6}, we get $\gamma_i({\theta_i})(\theta-\hat{\theta})\ge \gamma_i(\hat{\theta_i})(\theta-\hat{\theta})$,
	\begin{comment}
	\begin{align*}
	%	\label{IC_add}
	& \gamma_i({\theta_i})(\theta-\hat{\theta})\ge \gamma_i(\hat{\theta_i})(\theta-\hat{\theta}).
	\end{align*}
	\end{comment}
	which implies monotonicity of $\gamma_i({\theta_i})$ (constraint \eqref{cons1}). To derive Equation \eqref{cons2}, we can rearrange Equation  \eqref{IC_6} as follows:
	\begin{align}
	\label{IC_form}
	& V_i(\hat{\theta_i})+\theta_i\gamma_i(\hat{\theta_i})-V_i({\theta_i})-\theta_i\gamma_i({\theta_i}) \le r_i({\theta_i})-r_i(\hat{\theta_i}) \\
	\nonumber
	& \le V_i(\hat{\theta_i})+\hat{\theta_i}\gamma_i(\hat{\theta_i})-V_i({\theta_i})-\hat{\theta_i}\gamma_i({\theta_i}) .
	\end{align}	
	By considering $\hat{\theta}=\theta+
	\epsilon$ and dividing the entire Equation \eqref{IC_form} by $\epsilon\rightarrow 0$, we get:
	\begin{align}
	%\label{IC_form_d}
	\nonumber
	&\theta_i\frac{d}{d\theta_i} \gamma_i({\theta_i}) - \frac{d}{d\theta_i} V_i({\theta_i}) \le \frac{d}{d\theta_i} r_i(\theta_i) \le %\\
	%\nonumber
	%&
	\theta_i\frac{d}{d\theta_i} \gamma_i({\theta_i}) - \frac{d}{d\theta_i} V_i({\theta_i}).
	\end{align}	
	This results in
	\begin{align}
		\label{IC_form_d2}
	\theta_i\frac{d}{d\theta_i} \gamma_i({\theta_i}) - \frac{d}{d\theta_i} V_i({\theta_i})  = \frac{d}{d\theta_i} r_i(\theta_i).
	\end{align} 
	Integrating Equation \eqref{IC_form_d2} with respect to $\theta_i$ %from $\underline{\theta}$ to ${\theta_i}$ 
	results in:
	\begin{equation}
	\begin{aligned}
	\label{RR}
	r_i(\theta_i)-r_i(\underline{\theta})=& \int_{\underline{\theta}}^{\theta_i}y\frac{d}{d\theta_i}  \gamma_i({\theta_i})\bigg|_{\theta_i=y}+\frac{d}{d\theta_i}  V_i({\theta_i})\bigg|_{\theta_i=y} dy.
	\end{aligned}
	\end{equation}
%	Let us consider $\tilde{U}^{NU}_{i}(\underline{\theta_i},\underline{\theta_i}) = 0$ 
Using Proposition \ref{Lem1},
	the integration by parts of \eqref{RR} gives us: 
	\begin{align*}
	r_{i}({\theta_i})=   [\int_{\underline{\theta}}^{\theta_i} \gamma_i  (y) dy] -\theta_i \gamma_i(\theta_i)- V_i(\theta_i).
	\end{align*}
\end{proof}
%\subsection{Mechanism solution: demand content}
%\label{Sec:Solution of Mechanism}
In the following, we derive the optimal solution for the content demand consumption.
\begin{comment}
we show how to formulate the CP's utility (Equation \ref{eq:utility_CP}) by exploiting the characterization of $r_i(\cdot)$ in Equation \eqref{cons2}:
%%%%%%%%%%
\begin{align}
\nonumber
U^{CP}=\mathbf{E}_{\theta}&\sum_{i\in\mathcal{N}}[sx_i(\theta) - \frac{t}{2}x_i^2(\theta)+\phi (\theta_i)x_i(\theta)\sum_{j\in\mathcal{N}} g_{ij}x_j(\theta)+\\
\label{UTP_new1}
&(a-p)x_i(\theta)-\frac{b}{2}x_i^2(\theta)]
\end{align}
\end{comment}
%%%%%%%%%
Given the fact that $\mathbf{E}[\cdot ]$ is a linear operator, Equation \eqref{eq:utility_CP} can be rewritten as follows:
\begin{align}
\label{UTP_new}
&U^{CP}=\sum_{i\in\mathcal{N}}\mathbf{E}_{\theta_i}[C_i(\theta_i)-r_i(\theta_i)].
\end{align}
By substituting Equation \eqref{cons2} into Equation \eqref{UTP_new}, we obtain:
\begin{align*}
%\label{UTP_new2}
& U^{CP} =  \\
\nonumber
&\sum_{i\in\mathcal{N}} \int_{\underline{\theta}}^{\bar{\theta}} [ C_i(\theta_i)-\int_{\underline{\theta}}^{\theta_i}\gamma_{i}(s)ds +  V_i({\theta_i}) + \theta_i\gamma_i({\theta_i})]f(\theta_i)d\theta_i. 
\end{align*}
We can rewrite the term $\int_{\underline{\theta}}^{\bar{\theta}} \int_{\underline{\theta}}^{\theta_i}[\gamma_{i}(s)ds]f(\theta_i)d\theta_i$ as %follows:
\begin{comment}
\begin{align*}
%\label{hazard}
&	\int_{\underline{\theta}}^{\bar{\theta}} \int_{\underline{\theta}}^{\theta_i}[\gamma_{i}(s)ds]f(\theta_i)d\theta_i = \int_{\underline{\theta}}^{\bar{\theta}}\gamma_{i}(\theta_i) \frac{1-F(\theta_i) }{f(\theta_i)} f(\theta_i)d\theta_i,
\end{align*}
\end{comment}
$\int_{\underline{\theta}}^{\bar{\theta}}\gamma_{i}(\theta_i) \frac{1-F(\theta_i) }{f(\theta_i)} f(\theta_i)d\theta_i$,
where $F(\cdot)$ is the cumulative distribution function of $f(\cdot)$.
\begin{comment}
\begin{definition}
	We define $h(\theta_i) \equiv  \frac{f(\theta_i) }{1-F(\theta_i)}$, that captures the probability at which an event is expected to occur at a time, considering that it has not taken place yet \cite{nisanalgorithmic}.
\end{definition}
\end{comment}
%%%%%%%%%%%%%%%%%%%%%%%%%%%%%%%%%%%%%%%%
%$\phi({\theta_i})\equiv {\theta_i}-1/h({\theta_i}) \ge 0$
We define $h(\theta_i) \equiv  \frac{f(\theta_i) }{1-F(\theta_i)}$ and $\phi({\theta_i})\equiv {\theta_i}-1/h({\theta_i})$.
%Let $h(\theta_i) \equiv  \frac{f(\theta_i) }{1-F(\theta_i)}$ denote the hazard rate and $\phi({\theta_i})\equiv {\theta_i}-1/h({\theta_i})$ be the virtual private information of agent $i$ \cite{nisanalgorithmic}.
%%%%%%%%%%%%%%%%%%%%%%%%%%%%%%%%%%%%%%%%%
\begin{comment}
\begin{assumption}
	\label{assump2}
	$\phi(\underline{\theta})\equiv \underline{\theta}-1/h(\underline{\theta}) \ge 0$ satisfies $\phi({\theta_i}) \ge 0 $, for each $\theta_i  \in  [\underline{\theta}, \bar{\theta}]$ \cite{jadbabaie2019optimal} .
\end{assumption}
\end{comment}
%The term $\phi({\theta_i})$ is usually referred to as a virtual type of $\theta_i$, and $\phi(\underline{\theta_i}) \ge 0$ is satisfied for many distribution functions.
%Using Equation \eqref{lem2_2}, and the fact that $\mathbf{E}[\cdot]$ is a linear operator, the optimization problem \eqref{opt_1} can be written as follows:
Thus, the optimization \eqref{opt_2} can be rewritten as follows:
\begin{comment}
\begin{equation}
\begin{aligned}
\max_{ {\{ x_i(\hat{\theta_i},\hat{\theta}_{-i})\}_{i=1}^N}} & \mathbf{E}_{\theta}\sum_{i\in\mathcal{N}}[sx_i(\theta) - \frac{t}{2}x_i^2(\theta)+\phi (\theta_i)x_i(\theta)\sum_{j\in\mathcal{N}} g_{ij}x_j(\theta)\\
&+(a-p)x_i(\theta)-\frac{b}{2}x_i^2(\theta)],\\
& s.t. \; \; \; \; 
\gamma_{i} ' (\theta_i) \ge 0.
\label{opt_3}
\end{aligned}
\end{equation}
\end{comment}
\begin{equation}
\begin{aligned}
\max_{  x(\cdot)} & \mathbf{E}_{\theta}\sum_{i\in\mathcal{N}}[C_i(\theta_i)+V_i(\theta_i)+\phi(\theta_i)\gamma_{i} (\theta_i)]\\
& s.t. \; \; \; \; 
\gamma_{i} ' (\theta_i) \ge 0.
\label{opt_3}
\end{aligned}
\end{equation}
We make the following assumptions to guarantee the boundedness of the content demand of each user.
\begin{assumption}
	\label{assump1}
	For each $\theta_i  \in  [\underline{\theta}, \bar{\theta}]$ , $h(\theta_{i})$ is increasing\footnote{\cite{boyd2004convex} showed that $h(\theta_{i})$ is increasing for many common probability density functions, including normal and uniform distributions.} and $\phi({\theta_i}) \ge 0$ \cite{nisanalgorithmic}.% which results in $\phi({\theta_i}) \ge 0 $, for each $\theta_i  \in  [\underline{\theta_i}, \bar{\theta_i}]$  .
\end{assumption}
\begin{assumption}
	\label{asm1}
	For each $i \in \mathcal{N} $,  $t+b> (\bar{\theta} \sum_{j\in\mathcal{N},j \neq i}(g_{ij}+g_{ji})) $ and $s+a>p$.
	\begin{comment}
	\begin{itemize}
		\item 	 $t+b> (\bar{\theta} \sum_{j\in\mathcal{N},j \neq i}(g_{ij}+g_{ji})) $.
		\item $s+a>p$.
	\end{itemize}
		\end{comment}
\end{assumption}
In the following, we present a lemma that will be useful to characterize the optimal mechanism.
\begin{lemma}
	\label{Lem3}
	Let  us define matrix $K =[(t+b)I_N-(M_{\phi}G+G^TM_{\phi})]^{-1}$, where $M_{\phi} \equiv \text{diag}(\phi (\theta_1),\phi (\theta_2),...,\phi (\theta_N))$, $G = [g_{ij}], i, j \in \mathcal{N}$ and $I_N$ is the $N\times N$ identity matrix.
	Then, $\frac{\partial K}{\partial \theta _i}$ is a matrix with non-negative entries.% for any $i \in \mathcal{N}$.
\end{lemma}
\begin{proof}
	By applying the chain rule, we get:
	\begin{align}
	\label{KK-1}
	0=\frac{\partial KK^{-1}}{\partial \theta _i}=\frac{\partial K}{\partial \theta _i}K^{-1}+K\frac{\partial K^{-1}}{\partial \theta_i}.
	\end{align}
	Furthermore, $	\frac{\partial K^{-1}}{\partial \theta_i}=
	-(E_iG+G^TE_i)	$,
	\begin{comment}
		\begin{align*}
	%	\label{K-1}
	&	\frac{\partial K^{-1}}{\partial \theta_i}=%\frac{\partial }{\partial \theta _i}[(t+b)I_N-(M_{\phi}G+G^TM_{\phi})]\\
	%\nonumber
	%&=	-\frac{\partial }{\partial \theta _i}[M_{\phi}G+G^TM_{\phi}]=
	-(E_iG+G^TE_i),	
	\end{align*}
	\end{comment}
	where $E_i=\frac{\partial M_{\phi} }{\partial \theta _i}$ is a matrix with $\frac{\partial \phi (\theta_i)}{\partial \theta _i}$  at the $ii^{th}$ entry, and zero otherwise. Hence, using Equation \eqref{KK-1} we obtain:
	\begin{align}
	\label{K}
	\frac{\partial K}{\partial \theta_i}=-K\frac{\partial K^{-1}}{\partial \theta_i}K=K(E_iG+G^TE_i)K.
	\end{align}	
Due to Assumption \ref{assump1}, we have $\frac{\partial \phi (\theta_i)}{\partial \theta _i} \ge 0$. Thus,
	since the right-hand side of Equation \eqref{K} is a matrix with the non-negative entries, $\frac{\partial K}{\partial \theta _i}$ is a matrix with the non-negative entries as well.
\end{proof}
%%%%%%%%%%%%%%%%%%%%%%%
The next proposition characterizes the content demand function in terms of the optimal mechanism.
\begin{proposition}
The optimal mechanism defines as 
\begin{align}
\label{optX}
	{x}(\theta)=(s+a-p)[(t+b)I_N-(M_{\phi}G+G^TM_{\phi})]^{-1}\textbf{1}_{N,1},
	\end{align}
	where $\textbf{1}_{N,1}$ is the $N \times 1$ matrix, all of whose elements are equal to $1$.
\end{proposition}
\begin{proof}
	To find the solution of the optimization problem \eqref{opt_3}, we first ignore the monotonicity constraint. %$(\gamma_{i} ' (\theta_i) \ge 0)$
	 Later, we show that this constraint is indeed satisfied. %To find the optimal solution to \eqref{opt_3}, we maximize it point-wise. 
	Let  $\theta_{i} \in [\underline{\theta}, \bar{\theta}]$ be fixed and given. Hence, ${\{ x(\theta)\}}$ solves the following problem:
	\begin{align}
	\label{opt_4}
	\max_{  x(\theta)} &\sum_{i\in\mathcal{N}}[C_i(\theta_i)+V_i(\theta_i)+\phi(\theta_i)\gamma_{i} (\theta_i)]
	\end{align}
	The Hessian of the objective in optimization \eqref{opt_4} is given as:
	\begin{align*}
	%\label{hessian}
	\begin{pmatrix}
	-t-b & \cdots &\phi (\theta_1)g_{1N}+\phi (\theta_N)g_{N1}  \\
	%\phi (\theta_1)g_{12}+\phi (\theta_2)g_{21} &  \dots &\phi (\theta_2)g_{2N}+\phi (\theta_N)g_{N2}  \\
	\vdots & \ddots & \vdots\\
	\phi (\theta_1)g_{1N}+\phi (\theta_N)g_{N1}   & \dots &-t-b
	\end{pmatrix}
	\end{align*}
	Considering Assumption \ref{asm1}, this matrix is Hermitian and strictly diagonally dominant. Thus, it is negative semi-definite \cite{horn2012matrix}. Therefore, the objective of optimization \eqref{opt_4} is concave.
	The first order optimality condition of optimization \eqref{opt_4} yields:
	\begin{align}
	\label{frst_ordr}
	&(s+a-p)-(t+b)x_i(\theta)+\phi (\theta_i)\sum_{j\in\mathcal{N}}(g_{ij}x_j(\theta))\\
	\nonumber
	&+\sum_{j\in\mathcal{N}}(\phi (\theta_j)g_{ji}x_j(\theta))=0.
	\end{align}
%	Let us define $M_{\phi} \equiv \text{diag}(\phi (\theta_1),\phi (\theta_2),...,\phi (\theta_n))$. Thus,
Equation \eqref{frst_ordr} can be written in the matrix form as follows:
	\begin{align*}
	%\label{frst_ordr_mat}
	&(s+a-p){\textbf{1}}_{N,1}+(M_{\phi}G+G^TM_{\phi}){x}(\theta)=(t+b){x}(\theta).
	\end{align*}
%	Hence,
%	\begin{align*}
	% \label{frst_ordr_mat2}
%	&(s+a-p)\textbf{1}_{N,1}=[(t+b)I_N-(M_{\phi}G+G^TM_{\phi}]\textbf{x}(\theta).
%	\end{align*}
%	where $\textbf{1}_{N,P}$ denotes the $N \times P$ matrix whose all elements are equal to $1$.\\
	Since $[(t+b)I_N-(M_{\phi}G+G^TM_{\phi})]$ is a strictly diagonally dominant matrix, it is invertible. Hence,
	\begin{align*}
	&{x}(\theta)=(s+a-p)[(t+b)I_N-(M_{\phi}G+G^TM_{\phi})]^{-1}\textbf{1}_{N,1}.
	\end{align*}
	To finalize the proof, it is left to show that $ \gamma_{i}'=\frac{\partial \gamma_{i}(\theta_i)}{\partial \theta_i} \ge 0 $.
	%%%%%%%%%%%%%%%%%%%%%%%
Given the definition of $ \gamma_{i}$ in Equation \eqref{taghyir_motoghayer}, we obtain:
	\begin{align}
\label{gamma1}
	\frac{\partial \gamma_i  (\theta_i)}{\partial \theta _i}=
	E_{\theta_{-i}}&[\frac{\partial x_i(\theta_i,\theta_{-i})}{\partial \theta_i}\sum_{j\in\mathcal{N}} g_{ij}x_j(\theta_i,\theta_{-i}) +  \\
	\nonumber
	&x_i(\theta_i,\theta_{-i})\sum_{j\in\mathcal{N}} g_{ij}\frac{\partial x_j(\theta_i,\theta_{-i}) }{\partial \theta _i}]\ge 0. 
	\end{align}
		According to Lemma \ref{Lem3}, we have $\frac{\partial x_i(\theta_i)}{\partial \theta_i} \ge 0$ and $\frac{\partial x_j(\theta_i)}{\partial \theta_i} \ge 0$. Thus, by considering Lemma \ref{Lem3} and Equation \eqref{gamma1}, we can conclude that \textbf{$ \gamma_{i}' \ge 0 $.}
\end{proof}
{
\begin{remark}
\label{rmrk1}
The calculation of the optimal content demand from Equation \eqref{optX} has a complexity of $O(N^3)$, where N is the number of users in the network.
\end{remark}
}
\section{Case study}
In this section, we investigate the performance of our proposed mechanism in a case study and evaluate the effect of the network structure on the utility of the users and the CP. We set $s=1$, $t=1$, $a=0.5$, $b=6$, and $p=0.1$ as the parameters of the CP and the users. In addition, we assume  $g_{ij}=1 $ when agents $i$ and $j$ are connected, otherwise $g_{ij}=0 $. 

In the first step, we assess the validity of individual rationality and incentive compatibility of the proposed mechanism on a simple example. We consider a case with five users who can communicate through a fully connected graph (Figure \ref{fig:A}). In this network, four users share their private information %(which is the parameter that controls the strength of the network effect for user)
truthfully and User $5$ can share his private information untruthfully. The utility of User $5$ when he shares values of private information different from those of his actual private information is shown in Figure \ref{fig:IC}. Each curve in Figure \ref{fig:IC} corresponds to the different values of actual private information of User $5$, while the private information of the other four agents is fixed in all the curves.
%In this case, the CP sends the incentive reward to User $5$ to encourage him to send his truthful private information.
%This is also known as the IC property.
%Figure \ref{fig:IC} demonstrates the effectiveness of the proposed mechanism.
This figure shows that User $5$ gains his maximum utility when he shares his private information truthfully, which is marked by the black stars on the curves. Therefore, the mechanism satisfies the incentive compatibility constraint. 
%The black stars marked on the curves are the points where User $5$ gains maximum utility at each private information. According to this figure, User $5$ gains his maximum utility when he announces his private information truthfully, hence the mechanism satisfies incentive compatibility constraint. 
In addition, the utility of User $5$ when he shares his information truthfully is positive, which implies that the individual rationality is satisfied as well.
\begin{figure}
     \centering
     \begin{subfigure}[b]{0.2\textwidth}
         \centering
         \includegraphics[width=1\textwidth]{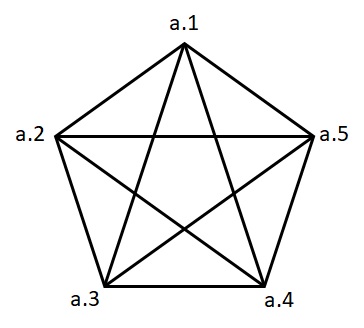}
         \caption{Fully connected}
         \label{fig:A}
     \end{subfigure}
     \begin{subfigure}[b]{0.2\textwidth}
         \centering
         \includegraphics[width=1\textwidth]{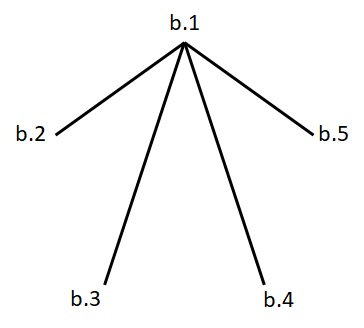}
         \caption{Star }
         \label{fig:B}
     \end{subfigure}
        %\begin{subfigure}[b]{0.15\textwidth}
       %  \centering
      %   \includegraphics[width=1\textwidth]{plot/graphC}
     %    \caption{Random}
    %     \label{fig:C}
   %  \end{subfigure}
        \caption{The bidirectional network structure between different users.}
        \label{fig:graphs}
\end{figure}
\begin{figure}[h!]
	\centering
	\includegraphics[width=1\linewidth]{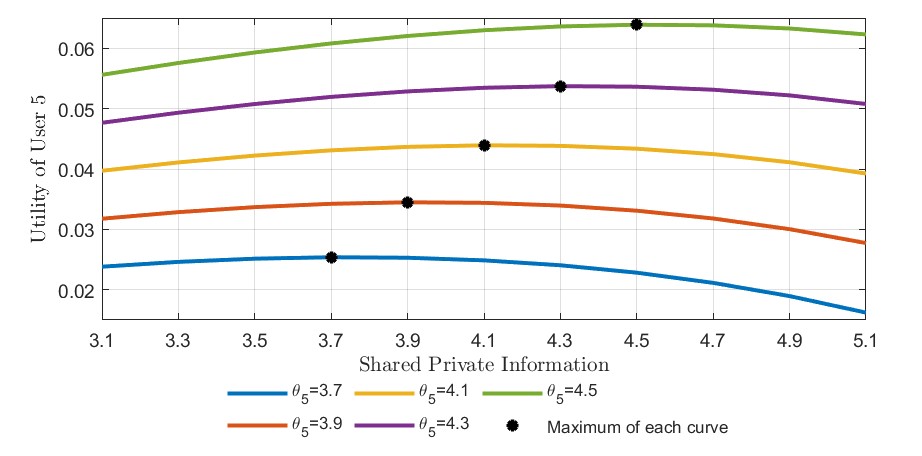}
	\caption{Utility of User $5$  when sharing different levels of private information.
Each curve corresponds to the different levels of the actual private information (i.e. $\theta_{5}$) of User $5$. 
%	including $\theta_{5}=3.7, 3.9, 4.1, 4.3, 4.5$
	%The star shows the maximum utility point.
	}
	\label{fig:IC}
\end{figure}

In the next step, we evaluate the impact of the network architecture on the utility  %and allocated content 
of the users. %through which the agents are connected on the reward function.
We consider three kinds of users, including a fully connected, a branch, and a central user. Two types of networks are considered for this evaluation: a fully connected network (Figure \ref{fig:A}) and a star network (Figure \ref{fig:B}). 
We assume that the users in both of these networks have the same parameters and private information and differ only in terms of their connectivity. Thus, the  utilities of all the agents in Figure \ref{fig:A} are the same. In this case, to investigate the behavior of fully connected users, we consider as an example the utility of user a.1, which is the same as all other agents in Figure \ref{fig:A}. In addition, since users $b.2,b.3,b.4$, and $b.5$ have the same parameters as well as the same connection, we only consider the utility of user b.2 in order to investigate the attributes of branch-users. Also, we consider agent b.1 to be the central agent.
%The utility and the allocated content of
 As Figure \ref{fig:compare} demonstrates, fully connected users have the highest utility compared to the other two types of users since the mutual influence between the users in such a network is the highest. % This allocation encourages the other users to request more demand and gain more profit.
 Moreover, the utility of the central user is higher than that of the branch user in the star network since the central user has a greater mutual influence on other users in this network as compared to the branch user. Hence, we can conclude that the greater the influence of the user on other users, the more content that is allocated to that user (and, consequently, more utility). Please note that we also evaluated the strength of the connectivity between the users. All users are assumed to have the same strength of connectivity. 
\begin{comment}

\begin{figure}[h!]
	\centering
	\includegraphics[width=0.45\linewidth]{plot/graph2}
	\caption{The bidirectional network structure between different users. A) Fully connected graph, B) Star graph}
	\label{fig:graph2}
\end{figure}
\begin{figure}[h!]
	\centering
	\includegraphics[width=0.8\linewidth]{plot/compare2}
	\caption{Utility and allocated content to User A.1 (fully connected), B.1 (central user in the star network), and B.2 (branch user).}
	\label{fig:compare}
\end{figure}
\end{comment}
\begin{figure}
    % \centering
   %  \begin{subfigure}[b]{0.5\textwidth}
        % \centering
         \includegraphics[width=0.5\textwidth]{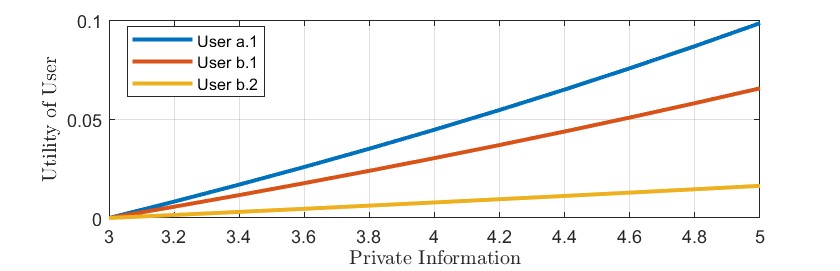}
      %   \caption{Utility of different users}
      %   \label{fig:utility}
    % \end{subfigure}
   %  \begin{subfigure}[b]{0.5\textwidth}
   %      \centering
  %       \includegraphics[width=1\textwidth]{plot/allocated}
%         \caption{Allocated content demand of different users }
  %       \label{fig:allocated}
    % \end{subfigure}
        \caption{Utility of Users a.1 (fully connected user), b.1 (central user), and b.2 (branch user).}
        \label{fig:compare}
\end{figure}

In the third step, we investigate the impact on the CP's utility of a case in which a user shares his private information untruthfully. The bidirectional network structure for this case study is depicted in Figure \ref{fig:graph1}. The utility of the CP when one user shares his private information untruthfully is displayed in Table \ref{tab:table1}. The result when all users share their information truthfully is displayed in the first row. %We assume that the truthful type is $3.2$ and each user announces $3.4$ as untruthful information. 
In the other rows, it is assumed that only one user shares the information untruthfully. As we can see, the untruthfully shared private information of User $1$, who is connected to all other users, has the highest impact on the CP's utility. In contrast, Users $2$ and $5$, who only connect to one user, have the lowest effect on the CP's utility. Hence, we conclude that the CP should prioritize designing the mechanism such that it guarantees incentive compatibility for the user with the greatest influence on others.% since he has more beneficial for him.

%\begin{comment}
\begin{figure}[!ht]
	\centering
	\includegraphics[width=0.4\linewidth]{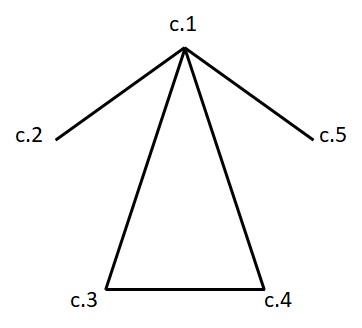}
	\caption{The bidirectional network structure between different users.}
	\label{fig:graph1}
\end{figure}
%\end{comment}

\begin{table}
	\caption{ 
		Impact of an untruthful user on the utility of CP}
	\centering
	\begin{tabular}[b]{c c}
		\hline
		Untruthful user  & Utility of CP  \\
		\hline
		No of the user & 83.37\\
		User c.1 & 27.23\\ 
		User c.2  & 64.93\\
		User c.3  & 38.84\\
		User c.4  & 38.84\\
		User c.5  & 64.93\\
	%%	User 6  & 38.84\\
		\hline
	\end{tabular}
	\label{tab:table1}
\end{table}
{

In the last evaluation step, the scalability and robustness of the proposed algorithm are demonstrated through tests on various network sizes, including large networks with many users. In each network, it is assumed that users connect randomly to half of the total users. For simplicity, we also assume that the users in all networks have the same parameters, and the same strength of connectivity. 
With a maximum network size of 800 users, each user is connected to 400 users, which is similar to the average number of friends per user on Facebook as reported by Pew Research Center \cite{Pew}. The impact of the number of users on the calculation time for optimal content demand is displayed in Table \ref{tab:table2} . These results support our statement in Remark \ref{rmrk1}. Figure \ref{fig:compareN} displays the comparison of the utility of the users in networks with different numbers of users. As shown in this figure, the user in a network with more nodes has the highest utility compared to users in networks with fewer nodes. This is due to the reason that as the number of users in the network increases, the users influence more users.

\begin{table}
	\caption{ {
		The impact of the number of users on the calculation time for optimal content demand }
		}
	\centering
	\begin{tabular}[b]{c c}
		\hline
		Number of users  & Time (sec) \\
		\hline
		10 & 0.019\\
		20 & 0.022\\ 
		50  & 0.037\\
		100  & 0.054\\
		200  & 0.146\\
		400  & 0.707\\
		600  & 1.685\\
		800 & 3.653\\
		\hline
	\end{tabular}
	\label{tab:table2}
\end{table}

\begin{figure}
    % \centering
   %  \begin{subfigure}[b]{0.5\textwidth}
        % \centering
         \includegraphics[width=0.5\textwidth]{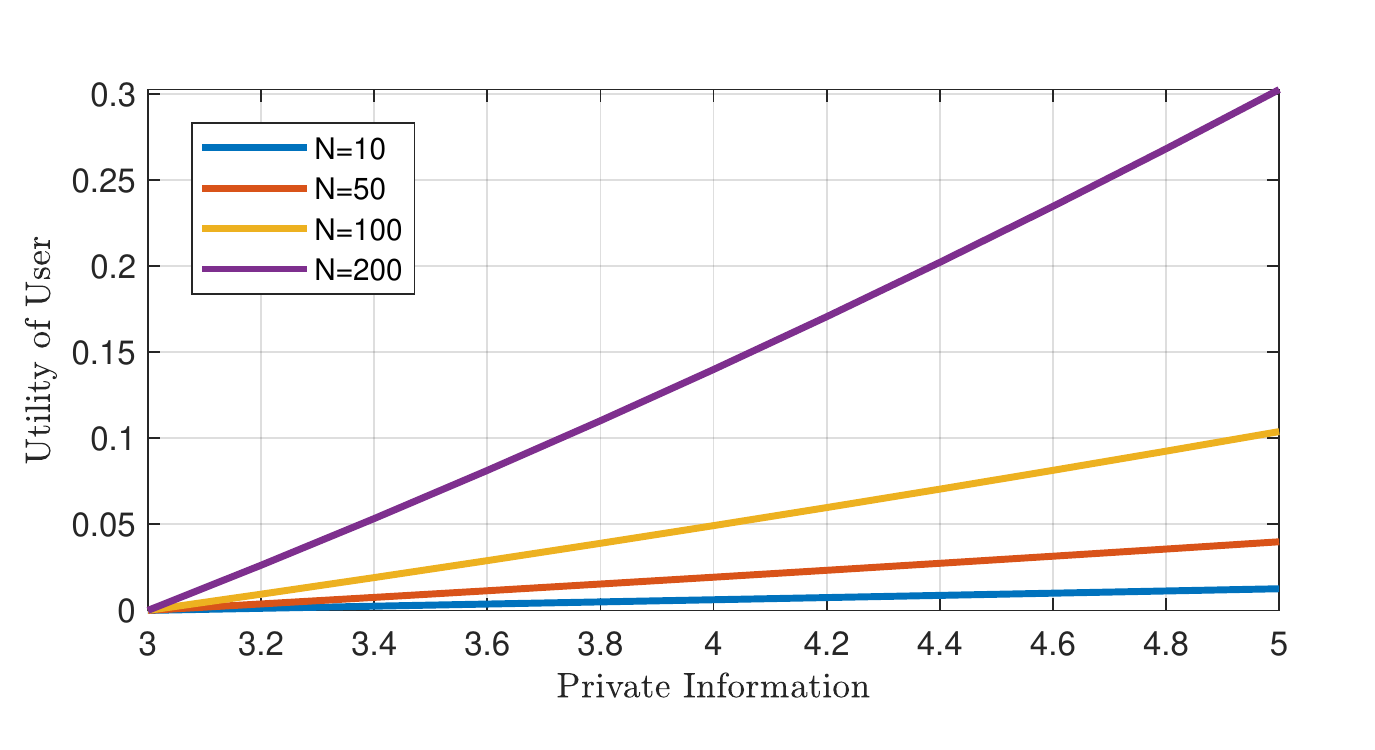}
      %   \caption{Utility of different users}
      %   \label{fig:utility}
    % \end{subfigure}
   %  \begin{subfigure}[b]{0.5\textwidth}
   %      \centering
  %       \includegraphics[width=1\textwidth]{plot/allocated}
%         \caption{Allocated content demand of different users }
  %       \label{fig:allocated}
    % \end{subfigure}
        \caption{{Utility of users in networks with different numbers of users.}}
        \label{fig:compareN}
\end{figure}

}

%In addition, we study the effect of fully connected network on the CP's utility function. In the case that there are $6$ Nus who connect through the fully connected graph and send their truthful type $3.2$ to the CP, the CP's utility is $1.99$ which is more less than the utility of CP for network structure of \ref{fig:graph1}. This result is due to the fact that the total incentive function to the NUs in the fully connected graph is $107.12$ that is higher than for the graph with dominant node (Node $1$ is dominant in \ref{fig:graph1}) which is $18.31$. In other words, CP incentives the dominant node more than others and then this node can affect the other nodes but in the fully connected all the nodes are the same and the CP should incentive all Nus the same which causes to increase the amount of incentive function and decrease the CP's utility function.

\section{Conclusion}
In this paper, we proposed an incentive mechanism for the sponsored content market by considering the social interaction of the agents. The strength of the social network ties of each user is considered to be his private information. We formulate the problem as the constrained functional optimization, which obtains the content demand and incentive reward functions as the CP's decisions. The optimal decisions satisfy incentive compatibility and individual rationality. The results in the case study verify the effectiveness and  scalability of the proposed mechanism. 

A future research direction would be to formulate the market as a mechanism whereby the users decide about their content demand and the content provider obtains the incentive reward.

\ifCLASSOPTIONcaptionsoff
  \newpage
\fi

% trigger a \newpage just before the given reference
% number - used to balance the columns on the last page
% adjust value as needed - may need to be readjusted if
% the document is modified later
%\IEEEtriggeratref{8}
% The "triggered" command can be changed if desired:
%\IEEEtriggercmd{\enlargethispage{-5in}}

% references section

% can use a bibliography generated by BibTeX as a .bbl file
% BibTeX documentation can be easily obtained at:
% http://mirror.ctan.org/biblio/bibtex/contrib/doc/
% The IEEEtran BibTeX style support page is at:
% http://www.michaelshell.org/tex/ieeetran/bibtex/
%\vspace*{-1mm}
\bibliographystyle{IEEEtran}
% argument is your BibTeX string definitions and bibliography database(s)
%\bibliography{IEEEabrv,../bib/paper}
%
\bibliography{main}

% Generated by IEEEtran.bst, version: 1.14 (2015/08/26)
\begin{thebibliography}{10}
\providecommand{\url}[1]{#1}
\csname url@samestyle\endcsname
\providecommand{\newblock}{\relax}
\providecommand{\bibinfo}[2]{#2}
\providecommand{\BIBentrySTDinterwordspacing}{\spaceskip=0pt\relax}
\providecommand{\BIBentryALTinterwordstretchfactor}{4}
\providecommand{\BIBentryALTinterwordspacing}{\spaceskip=\fontdimen2\font plus
\BIBentryALTinterwordstretchfactor\fontdimen3\font minus
  \fontdimen4\font\relax}
\providecommand{\BIBforeignlanguage}[2]{{%
\expandafter\ifx\csname l@#1\endcsname\relax
\typeout{** WARNING: IEEEtran.bst: No hyphenation pattern has been}%
\typeout{** loaded for the language `#1'. Using the pattern for}%
\typeout{** the default language instead.}%
\else
\language=\csname l@#1\endcsname
\fi
#2}}
\providecommand{\BIBdecl}{\relax}
\BIBdecl

\bibitem{yang2017noncooperative}
Y.~Yang, Z.~Lu, V.~O. Li, and K.~Xu, ``Noncooperative information diffusion in
  online social networks under the independent cascade model,'' \emph{IEEE
  Transactions on Computational Social Systems}, vol.~4, no.~3, pp. 150--162,
  2017.

\bibitem{xiong2020contract}
Z.~Xiong, J.~Zhao, Z.~Yang, D.~Niyato, and J.~Zhang, ``Contract design in
  hierarchical game for sponsored content service market,'' \emph{IEEE
  Transactions on Mobile Computing}, 2020.

\bibitem{wang2019optimization}
Y.~Wang, Z.~Cai, Z.-H. Zhan, Y.-J. Gong, and X.~Tong, ``An optimization and
  auction-based incentive mechanism to maximize social welfare for mobile
  crowdsourcing,'' \emph{IEEE Transactions on Computational Social Systems},
  vol.~6, no.~3, pp. 414--429, 2019.

\bibitem{hu2020blockchain}
J.~Hu, K.~Yang, K.~Wang, and K.~Zhang, ``A blockchain-based reward mechanism
  for mobile crowdsensing,'' \emph{IEEE Transactions on Computational Social
  Systems}, vol.~7, no.~1, pp. 178--191, 2020.

\bibitem{xiong2017economic}
Z.~Xiong, S.~Feng, D.~Niyato, P.~Wang, and Y.~Zhang, ``Economic analysis of
  network effects on sponsored content: a hierarchical game theoretic
  approach,'' in \emph{GLOBECOM 2017-2017 IEEE Global Communications
  Conference}.\hskip 1em plus 0.5em minus 0.4em\relax IEEE, 2017, pp. 1--6.

\bibitem{nie2018stackelberg}
J.~Nie, J.~Luo, Z.~Xiong, D.~Niyato, and P.~Wang, ``A stackelberg game approach
  toward socially-aware incentive mechanisms for mobile crowdsensing,''
  \emph{IEEE Transactions on Wireless Communications}, vol.~18, no.~1, pp.
  724--738, 2018.

\bibitem{nisanalgorithmic}
N.~Nisan, T.~Roughgarden, E.~Tardos, and V.~V. Vazirani, ``Algorithmic game
  theory,'' \emph{Google Scholar Digital Library Digital Library}, 2007.

\bibitem{zhang2017survey}
Y.~Zhang, M.~Pan, L.~Song, Z.~Dawy, and Z.~Han, ``A survey of contract
  theory-based incentive mechanism design in wireless networks,'' \emph{IEEE
  wireless communications}, vol.~24, no.~3, pp. 80--85, 2017.

\bibitem{roughgarden2010algorithmic}
T.~Roughgarden, ``Algorithmic game theory,'' \emph{Communications of the ACM},
  vol.~53, no.~7, pp. 78--86, 2010.

\bibitem{cai2020truth}
Y.~Cai, G.~Fragkos, E.~E. Tsiropoulou, and A.~Veneris, ``A truth-inducing sybil
  resistant decentralized blockchain oracle,'' in \emph{2020 2nd Conference on
  Blockchain Research \& Applications for Innovative Networks and Services
  (BRAINS)}.\hskip 1em plus 0.5em minus 0.4em\relax IEEE, 2020, pp. 128--135.

\bibitem{cai2022truthful}
Y.~Cai, N.~Irtija, E.~E. Tsiropoulou, and A.~Veneris, ``Truthful decentralized
  blockchain oracles,'' \emph{International Journal of Network Management},
  vol.~32, no.~2, p. e2179, 2022.

\bibitem{mirrlees1997information}
J.~A. Mirrlees, ``Information and incentives: The economics of carrots and
  sticks,'' \emph{The Economic Journal}, vol. 107, pp. 1311--1329, 1997.

\bibitem{montazeri2021optimal}
M.~Montazeri, H.~Kebriaei, B.~N. Araabi, and D.~Niyato, ``Optimal mechanism
  design in the sponsored content service market,'' \emph{IEEE Communications
  Letters}, vol.~25, no.~9, pp. 3051--3054, 2021.

\bibitem{jadbabaie2019optimal}
A.~Jadbabaie and A.~Kakhbod, ``Optimal contracting in networks,'' \emph{Journal
  of Economic Theory}, vol. 183, pp. 1094--1153, 2019.

\bibitem{montazeri2022distributed}
M.~Montazeri, H.~Kebriaei, B.~N. Araabi, J.~Kang, and D.~Niyato, ``Distributed
  mechanism design in continuous space for federated learning over vehicular
  networks,'' \emph{IEEE Transactions on Vehicular Technology}, 2022.

\bibitem{xiong2019dynamic}
Z.~Xiong, D.~Niyato, P.~Wang, Z.~Han, and Y.~Zhang, ``Dynamic pricing for
  revenue maximization in mobile social data market with network effects,''
  \emph{IEEE Transactions on Wireless Communications}, vol.~19, 2019.

\bibitem{candogan2012optimal}
O.~Candogan, K.~Bimpikis, and A.~Ozdaglar, ``Optimal pricing in networks with
  externalities,'' \emph{Operations Research}, vol.~60, no.~4, pp. 883--905,
  2012.

\bibitem{boyd2004convex}
S.~Boyd and L.~Vandenberghe, \emph{Convex optimization}.\hskip 1em plus 0.5em
  minus 0.4em\relax Cambridge university press, 2004.

\bibitem{horn2012matrix}
R.~A. Horn and C.~R. Johnson, \emph{Matrix analysis}.\hskip 1em plus 0.5em
  minus 0.4em\relax Cambridge university press, 2012.

\bibitem{Pew}
\BIBentryALTinterwordspacing
Branka. Facebook statistics – 2023. [Online]. Available:
  \url{https://truelist.co/blog/facebook-statistics/}
\BIBentrySTDinterwordspacing

\end{thebibliography}
\end{document}